\documentclass[sigconf]{acmart}

\usepackage{multirow}
\usepackage{multicol}
\usepackage{subfigure}
\usepackage{relsize}
\usepackage{enumitem}
\usepackage{amsmath}
\usepackage{tikz}
\usepackage{amsthm}
\usepackage{booktabs}
\usepackage{graphicx}
\usepackage{caption}
\usepackage{adjustbox}
\usepackage{pgfplots}
\usepackage{threeparttable}
\usepackage{ragged2e}
\usepackage{colortbl}
\newcommand{\change}[1]{\textcolor{black}{#1}}

\newcommand{\changenew}[1]{\textcolor{black}{#1}}

\newtheorem{assumption}{Assumption}
\usetikzlibrary{positioning}
\definecolor{blue1}{HTML}{AED4E5}
\definecolor{blue2}{HTML}{5795C7}
\definecolor{blue3}{HTML}{1E4C9A}

\definecolor{color1}{HTML}{737373}
\definecolor{color2}{HTML}{077df0}
\definecolor{color3}{HTML}{ff2f4c}
\newcolumntype{P}[1]{>{\centering\arraybackslash}p{#1}}
\newcolumntype{R}[1]{>{\RaggedLeft\arraybackslash}p{#1}}
\newcolumntype{L}[1]{>{\RaggedRight\arraybackslash}p{#1}}

\copyrightyear{2026}
\acmYear{2026}
\setcopyright{cc}
\setcctype{by}
\acmConference[KDD 2026] {Proceedings of the 32nd ACM SIGKDD Conference on Knowledge Discovery and Data Mining V.2}{August 9--13, 2026}{Jeju Island, Republic of Korea.}
\acmBooktitle{Proceedings of the 32nd ACM SIGKDD Conference on Knowledge Discovery and Data Mining V.2 (KDD 2026), August 9--13, 2026, Jeju Island, Republic of Korea}
\acmISBN{979-8-4007-2259-2/2026/08}
\acmDOI{10.1145/3770855.3818420}

\begin{document}

\title{From Scaling to Structured Expressivity: Rethinking Transformers for CTR Prediction}





\author{Bencheng Yan$^{*}$}
\email{bencheng.ybc@alibaba-inc.com}
\affiliation{%
  \institution{Alibaba Group}
  \city{Beijing}
  \country{China}
}

\author{Yuejie Lei$^{*}$}
\email{leiyuejie.lyj@alibaba-inc.com}
\affiliation{%
  \institution{Alibaba Group}
  \city{Beijing}
  \country{China}
}

\author{Zhiyuan Zeng$^{*}$}
\email{zengzhiyuan.zzy@alibaba-inc.com}
\affiliation{%
  \institution{Alibaba Group}
  \city{Beijing}
  \country{China}
}

\author{Zheye Deng$^{*}$}
\email{dengzheye.dzy@alibaba-inc.com}
\affiliation{%
  \institution{Alibaba Group}
  \city{Beijing}
  \country{China}
}

\author{Di Wang}
\email{zhemu.wd@alibaba-inc.com}
\affiliation{%
  \institution{Alibaba Group}
  \city{Beijing}
  \country{China}
}

\author{Kaiyi Lin}
\email{linkaiyi.lky@alibaba-inc.com}
\affiliation{%
  \institution{Alibaba Group}
  \city{Beijing}
  \country{China}
}

\author{Pengjie Wang}
\email{pengjie.wpj@alibaba-inc.com}
\affiliation{%
  \institution{Alibaba Group}
  \city{Beijing}
  \country{China}
}

\author{Chuan Yu}
\email{yuchuan.yc@alibaba-inc.com}
\affiliation{%
  \institution{Alibaba Group}
  \city{Beijing}
  \country{China}
}

\author{Jian Xu}
\email{xiyu.xj@alibaba-inc.com}
\affiliation{%
  \institution{Alibaba Group}
  \city{Beijing}
  \country{China}
}

\author{Bo Zheng$^{\dagger}$}
\email{bozheng@alibaba-inc.com}
\affiliation{%
  \institution{Alibaba Group}
  \city{Beijing}
  \country{China}
}

\thanks{$*$ These authors contributed equally to this work and are co-first authors.}
 \thanks{$\dagger$ Corresponding author}

\renewcommand{\shortauthors}{Bencheng Yan, et al.}

\begin{abstract}
Despite massive investments in scale, deep models for click-through rate (CTR) prediction often exhibit rapidly diminishing returns---a stark contrast to the \change{predictable scaling laws} seen in large language models (LLMs). We identify the root cause as a \change{fundamental} \textit{structural misalignment}: \change{standard} Transformers assume sequential compositionality, whereas CTR data demand combinatorial reasoning over \change{heterogeneous} fields. 
To restore alignment, we introduce the \textbf{Field-Aware Transformer (FAT)}. \change{By reconstructing the standard Transformer block with field-centric parameters, FAT achieves \textit{structured expressivity}, \changenew{fundamentally shifting the model complexity dependence from the total vocabulary size $n$ with the number of fields $F$ ($n \gg F$).}}
Crucially, to decouple model capacity from field cardinality, FAT employs a \change{{Basis-Composed Hypernetwork}} to synthesize field-specific parameters from shared bases, further reducing parameter complexity. \change{Theoretically, we ground this scaling behavior through a formal scaling law based on Rademacher complexity. Empirically, FAT outperforms exisiting state-of-the-art methods with up to \textbf{\changenew{+4.38\%}} AUC improvement, and delivers \textbf{+2.33\%} CTR and \textbf{+0.66\%} RPM in live production.} Our work establishes that scalable recommendation arises not from size alone, but from \textit{structured expressivity}---architectural coherence with data semantics.


\end{abstract}

\ccsdesc[500]{Information systems~Novelty in information retrieval}
\ccsdesc[500]{Information systems~Sponsored search advertising}
\ccsdesc[500]{Information systems~Personalization}

\vspace{-1em}
\keywords{Scaling Law,Recommendation Systems,CTR Prediction}

\maketitle

\vspace{-1em}
\section{Introduction}
\label{sec:introduction}

The success of large language models (LLMs) has revealed a powerful truth: when architecture and data are aligned, scaling becomes predictable. As model size, data volume, and compute increase, performance improves smoothly, enabling systematic progress~\cite{radford2018improving,radford2019language,DBLP:conf/acl/ZhaoQSTLMWM24,DBLP:journals/corr/abs-1807-03748,DBLP:journals/corr/abs-2001-08361}. This principle has inspired widespread efforts to transplant Transformer architectures into industrial recommendation systems, particularly for click-through rate (CTR) prediction~\cite{DBLP:journals/corr/abs-2311-05884,DBLP:conf/icml/ZhaiLLWLCGGGHLS24,DBLP:conf/icml/ZhangLCNLLZHYWP24,DBLP:conf/cikm/ZhuFZJWHDWZGYCC25,DBLP:conf/cikm/SunLWPLOJ19,DBLP:conf/icdm/KangM18,DBLP:conf/cikm/LiuHM24,yan2026unlocking}.
Yet, despite promising gains, most existing approaches remain at the level of \textbf{architectural mimicry}: they directly transplant LLM designs by tokenizing CTR features and applying standard Transformers, either within traditional pointwise prediction frameworks~\cite{DBLP:journals/corr/abs-2311-05884,DBLP:conf/icml/ZhangLCNLLZHYWP24,DBLP:conf/cikm/ZhuFZJWHDWZGYCC25,DBLP:journals/corr/abs-2409-12740} or under generative reformulations~\cite{DBLP:conf/icml/ZhaiLLWLCGGGHLS24,DBLP:journals/corr/abs-2410-02939,DBLP:conf/cikm/LiuHM24,yan2026unlocking}. While these methods benefit from increased model capacity, empirical studies have observed diminishing returns in performance as models scale up~\cite{DBLP:journals/corr/abs-2412-00714}. \change{This suggests a potential disconnect between architectural growth and effective learning: simply making the model larger does not guarantee it gets smarter.}
 
\change{We attribute this divergence to a \textbf{structural misalignment} between the inductive biases of standard Transformers and the nature of CTR data.} Although both modalities involve sequences of discrete tokens, their semantic structures are profoundly different:


\begin{itemize} [left=0pt]
    \item In natural language, \change{semantics are \textbf{sequential and compositional}: meaning arises from the hierarchical \changenew{construction} of \textit{homogeneous} tokens. Here, order is strict, and attention is driven by syntactic dependencies and context.}
    \item In CTR prediction, \change{predictive power is \textbf{combinatorial and field-dependent}:} user behavior is driven by specific cross-field conjunctions such as ``young user $\times$ luxury brand'' or ``mobile device $\times$ evening session''. Inputs are \textit{heterogeneous sets} of features, each belonging to a semantic \emph{field} (e.g., \texttt{user\_age}, \texttt{ad\_category}, \texttt{device\_type}), where order is arbitrary; what matters is the underlying interaction topology between distinct semantic fields.
\end{itemize}

Standard self-attention, designed for compositional semantics over dense, ordered sequences, fails to respect this distinction. It treats all embeddings uniformly via globally shared projection matrices.
Under extreme sparsity---where most field-value combinations are rarely observed—this unstructured attention risks amplifying noise and distorting gradients, which can hinder scalable learning. Even more troubling is the absence of a theoretical foundation for scaling in recommendation. While LLMs benefit from well-characterized generalization bounds and scaling laws grounded in statistical learning theory~\cite{DBLP:conf/acl/ZhaoQSTLMWM24,DBLP:journals/corr/abs-1807-03748,DBLP:journals/corr/abs-2001-08361}, 
\change{comparable theoretical frameworks remain largely underexplored for CTR models}. Without them, scaling becomes a trial-and-error process, disconnected from architectural design principles. 
This leads us to a pivotal question: \textit{Can we redesign the Transformer for recommendation such that its expressive capacity grows in harmony with the underlying interaction complexity of the data---not merely in raw parameter count, but in structured expressivity?}

To answer this, we draw upon a classical insight: \textbf{field-aware interaction modeling}. Models like Field-aware Factorization Machines (FFM)~\cite{DBLP:conf/recsys/JuanZCL16} assign dedicated latent vectors to each ordered field pair, enabling context-sensitive modeling.
\change{Inspired by this, a natural extension would be to make the Transformer's projection matrices \textit{specialized per field pair}. However, a direct realization incurs prohibitive parameter growth: with $F$ semantic fields, the number of interaction-specific parameters scales quadratically as $\mathcal{O}(F^2 d^2)$. In real-world systems where $F \sim 10^3$ and $d\sim 128$, even a moderate base model can balloon from 100 million to over \textbf{10 trillion parameters}, rendering deep scaling infeasible.}


To resolve this tension between expressivity and scalability, we introduce the \textbf{Field-Aware Transformer (FAT)}, \change{which achieves \textbf{structured expressivity} by reconstructing the standard Transformer block. Specifically, we introduce \textbf{Field-Decomposed Attention}, which factorizes the prohibitively expensive pair-wise transformation into two manageable components: (i) \textit{field-aware content alignment}, where \changenew{queries, keys, and values} are projected using matrices specific to their own fields (scaling as $\mathcal{O}(F d^2)$), and (ii) \textit{field-pair interaction modulation}, which governs information flow via lightweight scalars (scaling as $\mathcal{O}(F^2)$). Complementing this interaction modeling, we utilize a \textbf{Field-Aware FFN} to adapt intra-field feature distributions, ensuring distinct statistical properties of heterogeneous fields are respected. Additionally, we employ a \textbf{Basis-Composed Hypernetwork} to dynamically synthesize field-specific parameters from a compact set of shared bases. This design significantly reduces parameter complexity and enables agile feature evolution with no additional inference overhead.}

\change{Moreover, we provide a theoretical foundation for these results by deriving a \textbf{formal scaling law for CTR models} based on Rademacher complexity. Our analysis indicates that} FAT’s generalization error depends on the combinatorial structure of field interactions rather than the total vocabulary size. \change{This theoretical perspective offers a potential explanation for the {power-law scaling} in AUC observed in our experiments, contrasting with the saturation often seen in baseline Transformers.}


Our contributions are summarized as follows:
\begin{itemize} [left=0pt]
    \item \textbf{Methodology}: We propose \textbf{FAT}, which aligns the Transformer architecture with CTR via field-aware priors, while its hypernetwork implementation enables sustainable feature evolution.

    \item \textbf{Theoretical Analysis}: We derive a generalization bound based on Rademacher complexity. This provides a theoretical foundation for scaling in CTR, suggesting that performance depends on the combinatorial structure of interactions.
    
    \item \textbf{Empirical Validation}: On large-scale benchmarks, FAT sets a new state-of-the-art with up to \textbf{+4.38\%} AUC improvement. Deployed online, it delivers \textbf{+2.33\%} CTR and \textbf{+0.66\%} RPM, demonstrating significant business impact.
\end{itemize}
    
    
    


\vspace{-1em}
\section{Related Work}
\label{sec:related_work}


\subsection{\change{Structured Feature Interaction Modeling}}
\noindent\textbf{\change{Shallow Factorization Models.}}
Factorization Machines (FM)~\cite{DBLP:conf/icdm/Rendle10} laid the foundation for interaction modeling by introducing low-rank pairwise approximations. This concept was refined by Field-aware FM (FFM)~\cite{DBLP:conf/recsys/JuanZCL16}, which assigns field-pair-specific latent vectors to model context-sensitive effects. Field-weighted FM (FwFM)~\cite{DBLP:conf/www/PanXRZPSL18} was proposed to capture interaction importance via lightweight learnable scalar weights. However, regardless of their parameter efficiency, these models are inherently limited to shallow, second-order interactions.


\smallskip
\noindent\textbf{\change{Deep Interaction Architectures.}}
\change{To capture high-order patterns,} neural extensions~\cite{DBLP:conf/ijcai/GuoTYLH17,DBLP:conf/cikm/SongS0DX0T19,DBLP:conf/www/WangSCJLHC21,yan2022apg,li2021explicit} \change{integrate Multilayer Perceptrons (MLPs) or multi-head self-attention layers}. However, these methods typically restrict field-aware priors to the input layer. Consequently, structured interaction modeling remains confined to shallow architectures, preventing deep compositionality and making the principled scaling of structured interactions infeasible.



\subsection{Towards Predictable Scaling in Recommendation}

\noindent\textbf{\change{Transformer Adaptation and Misalignment.}}
Recent efforts adopt Transformer architectures for CTR prediction~\cite{DBLP:conf/icml/ZhangLCNLLZHYWP24,DBLP:conf/icml/ZhaiLLWLCGGGHLS24,DBLP:journals/corr/abs-2311-05884,DBLP:conf/icdm/KangM18,DBLP:conf/cikm/ZhuFZJWHDWZGYCC25,wang2025scaling,ye2025fuxi}.However, standard self-attention operates under assumptions---sequential order, dense tokens, compositional syntax---that do not hold in recommendation, where inputs are unordered, sparse sets with combinatorial semantics. Applying unstructured attention may lead to inefficient representation learning and poor generalization under sparsity.


\noindent\textbf{\change{Scaling Behavior and Bottlenecks.}}
While LLMs exhibit predictable power-law scaling where performance improves smoothly with model size and compute~\cite{DBLP:journals/corr/abs-2001-08361}, CTR models struggle to replicate this behavior. Empirical studies indicate that simply scaling up the width or depth of generic deep models often leads to diminishing returns or even performance degradation due to overfitting~\cite{DBLP:conf/icml/ZhaiLLWLCGGGHLS24}. This suggests a misalignment between architectural capacity and data structure, and current literature lacks a principled connection between model design and scalable behavior. \change{Consequently, a critical research gap persists: existing approaches typically force a trade-off between the fine-grained, interpretable priors of factorization models and the scalable capacity of deep neural networks. How to reconcile these conflicting objectives---achieving deep, predictable scaling without sacrificing domain-aware structural constraints---remains an open challenge in the field.}




\begin{figure*}[t]
\centering
\includegraphics[width = .9\textwidth]{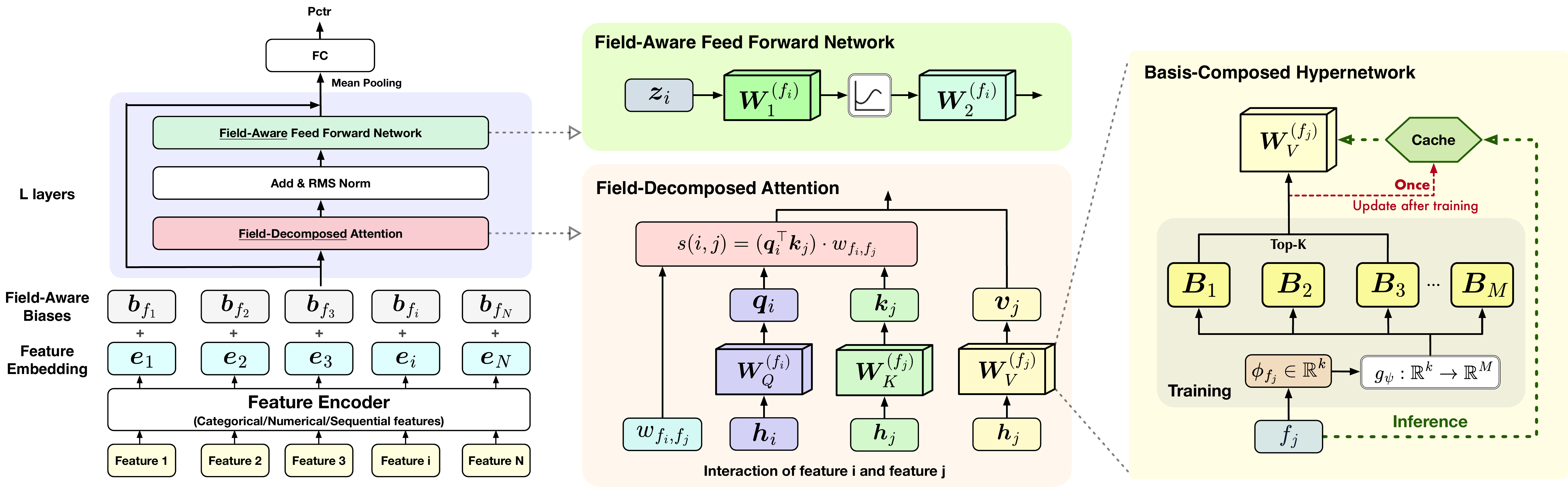}
\vspace{-1em}
\caption{\changenew{The architecture of Field-Aware Transformer (FAT).}}
\vspace{-1em}
\label{fig:fat_arch}
\end{figure*}

\vspace{-1.5em}
\section{Method}
\label{sec:method}

As illustrated in Figure~\ref{fig:fat_arch}, the Field-Aware Transformer (FAT) \change{enforces semantic consistency from input to output. Building on \textbf{Field-Aware Tokenization}}~(Section~\ref{sec:fa-tokenizer}), it achieves \emph{structured expressivity} by \change{refining the Transformer block with two field-centric components: (1) \textbf{Field-Decomposed Attention}~(Section~\ref{sec:fa-attention}) to govern cross-field interactions; and (2) \textbf{Field-Aware Feed Forward Network}~(Section~\ref{sec:fa-ffn}) to adapt intra-field feature distributions.} To decouple model capacity from field cardinality, FAT employs a \change{\textbf{Basis-Composed Hypernetwork}~(Section~\ref{sec:hypernetwork}) that synthesizes these field-specific parameters from shared bases. This design dramatically reduces parameter complexity from $\mathcal{O}(F^2 d^2)$ to $\mathcal{O}(F d^2 + F^2)$,} ensuring scalability and tighter generalization bounds under extreme sparsity. 



\vspace{-1em}
\subsection{Field-Aware Tokenization}
\label{sec:fa-tokenizer}

In CTR prediction, inputs are unordered sets of heterogeneous features drawn from distinct semantic fields (e.g., user, item, context). \change{To unify these inputs, we map each raw feature $x_i$ to a dense token $\boldsymbol{e}_i \in \mathbb{R}^d$ via type-specific encoders: embedding lookups for categorical and discretized numerical features, and sequence pooling (e.g., DIN~\cite{DBLP:conf/kdd/ZhouZSFZMYJLG18}) for user behavior lists.}

\change{Accordingly, given the non-sequential nature of these feature sets, we replace standard positional encodings with \emph{field-aware biases}. This injects structural priors based on semantic roles rather than arbitrary positions. For a token $i$ belonging to field $f_i$, the final input representation is defined as:}
    $\boldsymbol{h}_i = \boldsymbol{e}_i + \boldsymbol{b}_{f_i},$
\change{where $\boldsymbol{b}_{f_i} \in \mathbb{R}^d$ is a learnable bias vector unique to field $f_i$. This formulation ensures permutation invariance while grounding each token in its semantic context, establishing the foundation for the field-decomposed modeling in FAT.}

\vspace{-1.5em}
\subsection{Field-Decomposed Attention}
\label{sec:fa-attention}
\change{Standard self-attention operates on a global basis, ignoring the heterogeneous nature of CTR fields. While assigning distinct projection matrices to each field pair $(f_i,f_j)$ could capture these interactions, it results in a prohibitive $\mathcal{O}(F^2d^2)$ parameter complexity.}

\change{To resolve this, we propose Field-Decomposed Attention, which factorizes the attention mechanism into: (1) \textbf{Field-aware content alignment}: how well two tokens interact, given their respective semantic roles; (2) \textbf{Field-pair interaction modulation}: how strongly information should flow from one field to another.  }

\change{Notably, to enable deep semantic alignment, we assign independent projection matrices $\{\boldsymbol{W}_Q^{(f)}, \boldsymbol{W}_K^{(f)}, \boldsymbol{W}_V^{(f)}\}$ to each field $f$. The raw attention score is then defined as:}
\begin{equation}
    s(i,j)=\left(\boldsymbol{q}_i^\top \boldsymbol{k}_j\right) \cdot w_{f_i,f_j}\ ,\ \text{with } \boldsymbol{q}_i=\boldsymbol{h}_i\boldsymbol{W}_Q^{(f_i)},\  \boldsymbol{k}_j=\boldsymbol{h}_j\boldsymbol{W}_K^{(f_j)}
\end{equation}
{The first term refers \textbf{field-aware content alignment}. Unlike generic similarity, the field-specific projections allow the model to tailor the transformation based on the field's role. For instance, it can distinguish whether an \texttt{age} token should be encoded differently when acting as a user profile feature versus a context feature.
Complementing this, the scalar parameter $w_{f_i,f_j} \in \mathbb{R}$ executes \textbf{field-pair interaction modulation}. Acting as a gate, it governs the information flow strength between fields. This allows for fine-grained control; for example, a large $w_{\texttt{ad\_category}, \texttt{user\_behavior}}$ amplifies relevant user interests, while a small $w_{\texttt{ad\_category},\texttt{device\_type}}$ suppresses noise from irrelevant couplings.}
\change{Finally, we derive the layer output by aggregating the field-projected value vectors $\boldsymbol{v}_j = \boldsymbol{h}_j\boldsymbol{W}_V^{(f_j)}$:}
\begin{equation}
    \text{FAT-Attn}(\boldsymbol{H})_i = \sum_{j=1}^N \left(\frac{\exp\left(s(i,j) / \sqrt{d}\right)}{\sum_{k=1}^N \exp\left(s(i,k) / \sqrt{d}\right)} \right) \boldsymbol{v}_j
\end{equation}

\change{This attention mechanism provides three critical advantages:}
\begin{itemize}[left=0pt]
 \item \textbf{Hierarchical Field Awareness}: The content alignment operates at a \textbf{\change{coarse-grained per-field level}} to encode semantic roles, while interaction modulation functions at a \textbf{\change{fine-grained field-pair level}} to govern routing strength.
 
 \item \textbf{\change{Interpretability and Assymetry}}: \change{The learned scalars $w_{f_i,f_j}$ naturally model directional effects ($w_{f_i,f_j} \neq w_{f_j,f_i}$) and reveal interpretable interaction patterns and insights between fields.}
 \item \textbf{Efficiency and Scalability}: \change{FAT eliminates the need for full-rank field-pair matrices, collapsing the parameter complexity from $\mathcal{O}(F^2 d^2)$ to $\mathcal{O}(F d^2 + F^2)$. This yields a reduction of over 99\% in typical settings (e.g., \changenew{from 150B to 0.5B}), enabling scalable training without sacrificing semantic fidelity.}
\end{itemize}
\change{In summary, {Field-Decomposed Attention} establishes the foundation for structured expressivity by decoupling representation from routing. It scales interaction modeling not by adding unstructured capacity, but by deepening semantic resolution in a controlled way.}

\vspace{-0.5em}
\subsection{Field-Aware Feed Forward Network}
\label{sec:fa-ffn}
\change{Similarly, standard Feed Forward Networks (FFN) apply uniform transformations across all tokens, ignoring the distinct feature distributions of heterogeneous fields. To address this, we introduce a Field-Aware FFN. Let $\boldsymbol{z}_i \in \mathbb{R}^d$ denote the input representation for token $i$. Its transformation is explicitly conditioned on its field $f_i$:}
\begin{equation}
    \text{FAT-FFN}(\boldsymbol{z}_i)=\text{SiLU}\left(\boldsymbol{z}_i \boldsymbol{W}_1^{(f_i)}\right) \boldsymbol{W}_2^{(f_i)}
\end{equation}
\change{where $\boldsymbol{W}_1^{(f_i)} \in \mathbb{R}^{d\times d_{ff}}$ and $\boldsymbol{W}_2^{(f_i)} \in \mathbb{R}^{d_{ff}\times d}$ are field-specific projection matrices. Here, $d_{ff}$ is the intermediate expansion dimension.}

\subsection{Basis-Composed Hypernetwork}
\label{sec:hypernetwork}
\change{Despite its expressiveness, storing full-rank matrices for all fields (i.e., $\{\boldsymbol{W}_Q^{(f)}, \boldsymbol{W}_K^{(f)}, \boldsymbol{W}_V^{(f)}, \boldsymbol{W}_1^{(f)}, \boldsymbol{W}_2^{(f)}\}$) incurs prohibitive memory costs in large-scale systems. }
To decouple parameter growth from field cardinality while preserving semantic fidelity, we introduce a \emph{basis-composed hypernetwork} \change{to generate} field-specific projections.

\change{We construct {type-specific} basis sets to capture the distinct structural properties of each matrix. Let $\Phi \in \{Q, K, V, 1, 2\}$ denote the parameter type. For each type $\Phi$, we maintain a separate bank of $M$ basis matrices $\mathcal{B}_\Phi=\{\boldsymbol{B}_{1,\Phi},\cdots,\boldsymbol{B}_{M,\Phi}\}$, where each $\boldsymbol{B}_{m,\Phi}$ shares the identical dimensions as the target projection $\boldsymbol{W}_\Phi^{(f)}$.}
\change{For a given field $f$, a lightweight \changenew{single-layer MLP router} $g_\psi(\cdot)$} \changenew{maps the field embedding $\phi_f \in \mathbb{R}^{k}$ to a field-specific selection signature $s^{(f)} \in \mathbb{R}^M$}:
\begin{equation}
    s^{(f)} = g_\psi(\phi_f), \quad \pi_f = \text{top-}K(s^{(f)})
\end{equation}
\change{We then compute sparse blending coefficients via softmax over the selected indices:}
\vspace{-0.5em}
\begin{equation}
    \alpha_m^{(f)} = \frac{\exp\left( s_m^{(f)} \right)}{\sum_{m' \in \pi_f} \exp \left( s_{m'}^{(f)} \right)}, \quad \forall m \in \pi_f
\end{equation}
\change{Finally, the field-specific projection matrix for type $\Phi$ is synthesized by combining the corresponding bases:}
\begin{equation}
\boldsymbol{W}_\Phi^{(f)} = \sum\nolimits_{m \in \pi_f} \alpha_m^{(f)} \boldsymbol{B}_{m, \Phi}
\end{equation}



The hypernetwork enables scalable architecture design with three key advantages:
\begin{itemize}[left=0pt]
    \item \textbf{Decoupled complexity from field \change{cardinality}:}
    Field-specific projections are synthesized \change{from shared bases}, breaking the \change{prohibitively expensive} $\mathcal{O}(F d^2)$ parameter dependency. \change{Storage complexity grows linearly as $\mathcal{O}(M d^2 + F k)$, preventing parameter explosion even as the number of fields $F$ scales to thousands}.

    \item \textbf{\change{Sustainable feature evolution:}}
    \change{The design supports agile industrial iteration. Introducing new feature fields requires only initializing a lightweight meta-embedding $\phi_f$, rather than allocating full-rank matrices. This ensures continuous feature expansion with minimal storage overhead and rapid model adaptation.}
    
    \item \textbf{Zero \change{dynamic generation cost}:}
    \change{While dynamically defined during training, parameters are materialized offline for inference.} All \change{field-specific matrices} are precomputed and cached after training. \change{Consequently, this ensures zero inference overhead, meeting strict production constraints.}
\end{itemize}




\subsection{CTR Prediction}
\change{FAT stacks $L$ layers of Field-Aware Transformer blocks. Each block consists of a field-decomposed attention sub-layer and a field-aware FFN sub-layer, employing layer normalization and residual connections to facilitate deep model training. For the $\ell$-th layer, the update rule is:} 
\begin{equation}
    \boldsymbol{Z}^{(\ell+1)} = \text{FAT-FFN}\left(\text{LayerNorm}\left(\text{FAT-Attn}\left(\boldsymbol{Z}^{(\ell)}\right)\right)\right) + \boldsymbol{Z}^{(\ell)}.
    \label{eq:residual}
\end{equation}

\change{Finally, we employ sum-pooling to aggregate field representations into a global instance vector, followed by a sigmoid projection for CTR prediction:}
    $\hat{y} = \sigma\left(  \sum_{i=1}^N \boldsymbol{z}_i^{(L)} \boldsymbol{w}\right),$
\change{where $\sigma(\cdot)$ is the sigmoid function, $\boldsymbol{z}_i^{(L)}$ is the final output vector for token $i$, and $\boldsymbol{w} \in \mathbb{R}^d$ is a learnable weight vector.}

\section{\change{Theoretical Analysis: Generalization and Scaling}}
\label{sec:theory}

One of the most sought-after goals in recommendation systems is achieving predictable, smooth performance improvement as model size increases---a phenomenon well-documented in large language models through empirical scaling laws~\cite{DBLP:journals/corr/abs-2001-08361}. However, in CTR prediction, naive scaling often leads to performance saturation or degradation due to unstructured capacity growth that amplifies noise rather than signal. We argue that FAT enables \textit{principled} scaling by aligning architectural design with the combinatorial semantics of feature interactions. 
\change{In particular, the effective model complexity of a FAT attention layer is governed by the number of semantic fields $F$ and field-pair interaction structure (rather than the raw vocabulary size $n = \sum_{i=1}^F |\mathcal{V}_i|$, which can exceed $10^9$ in industrial settings).}
\change{This alignment yields tighter, schema-dependent generalization control and supports well-behaved scaling in practice. Our main theoretical result establishes a generalization bound for a single FAT attention layer; the complete proof is provided in Appendix~\ref{app:theory}.}

\begin{theorem}[Generalization Bound for FAT]
\label{thm:generalization}
Let $\mathcal{D}$ be a distribution over input sequences $\boldsymbol{H} = [\boldsymbol{h}_1, ..., \boldsymbol{h}_N]$ with $\|\boldsymbol{h}_i\|_2 \leq R$. Assume all parameter matrices ($\boldsymbol{W}^{(f)}_Q, \boldsymbol{W}^{(f)}_K, \boldsymbol{W}^{(f)}_V$) have Frobenius norm bounded by $B$, and all interaction scalars ($w_{f_i,f_j}$) are bounded by $B_w$. Let $m$ be the number of training samples, \change{and let $L_{\text{train}}$ and $L_{\text{gen}}$ denote the empirical and population risks under a bounded $1$-Lipschitz loss.} Then, with probability at least $1-\delta$, the generalization error $L_{\text{gen}}$ of a single FAT layer satisfies:
\[
L_{\text{gen}} \leq L_{\text{train}} + \mathcal{O}\left( \frac{\sqrt{F d^2 + F^2}}{\sqrt{m}} \cdot C(R, B, B_w, d) + \sqrt{\frac{\log(1/\delta)}{m}} \right),
\]
where $C(R, B, B_w, d) = \mathcal{O}(R^2 B^2 B_w \sqrt{d} + R B B_w)$ is a constant depending on the norm bounds and embedding dimension.
\end{theorem}

This bound highlights a key advantage over standard Transformers. \change{\change{Standard self-attention operates uniformly across all tokens, lacking the structural constraints to distinguish between semantic fields. This unconstrained flexibility forces the model to learn interaction patterns from scratch, making it highly susceptible to overfitting spurious correlations under extreme data sparsity.} In contrast, FAT's field-conditioned transformations depend on the number of semantic fields $F$ (and field-pair modulations $F^2$), confining the effective hypothesis space to $\mathrm{poly}(F)$. This explicitly decouples the learning complexity from the massive vocabulary size $n$.} Consequently, every parameter is shared across vast amounts of data within its field, dramatically improving statistical efficiency and generalization.

Critically, this tight generalization bound enables \textit{predictable scaling}. For a fixed data distribution and field schema $F$, increasing the embedding dimension $d$ (and thus the total parameter count $N_{\text{params}} \propto F d^2$) deepens the model's representational fidelity for both field-aware content alignment and field-pair interaction modulation. This allows the model to learn higher-rank, more nuanced interaction functions between fields, thereby systematically reducing the training error $L_{\text{train}}$. Because the architecture is aligned with the data's combinatorial structure, these additional parameters refine meaningful patterns instead of fitting noise.

Combining these two effects---reduced bias from enhanced expressivity and reduced variance from tight generalization---\change{we provide a theoretical justification for the observed scaling behavior. Specifically, as we scale the model width, the predictable drop in training error coupled with our tight bound supports the power-law trend seen empirically:}
\begin{equation}
\label{eq:scaling_law}
\Delta \text{AUC} \propto N_{\text{params}}^{\beta}, \quad \beta > 0,
\end{equation}
where $\Delta \text{AUC}$ is the performance gain relative to a baseline. 
Empirical validation of this law is presented in Section~\ref{sec:scaling}.

\section{Experiments}
\label{sec:experiments}

We conduct comprehensive experiments to rigorously evaluate the effectiveness, interpretability, scalability, and deployability of our proposed Field-Aware Transformer (FAT). 
We aim to answer five key research questions:

\begin{itemize}[left=0pt]
    \item \textbf{RQ1}: \change{Does FAT consistently outperform state-of-the-art baselines across different model scales and datasets?} 
    \item \textbf{RQ2}: \change{How do the field-aware components and the hypernetwork contribute to the model's effectiveness?} 
    \item \textbf{RQ3}: \change{Do the learned interaction patterns demonstrate semantic coherence and align with domain intuition?} 
    \item \textbf{RQ4}: \change{Does FAT exhibit a predictable power-law scaling trend consistent with our theoretical analysis?} 
    \item \textbf{RQ5}: \change{Is FAT computationally efficient for online serving and does it deliver substantial business gains?} 
\end{itemize}

\subsection{Experimental Setup}

\subsubsection{Dataset}

\change{We conduct evaluations on two complementary datasets to validate both industrial applicability and open-source reproducibility: (1)
    \textbf{Taobao}, a large-scale CTR dataset from {Taobao's sponsored search}, containing over \textbf{14 billion user impressions} collected over two weeks. 
    (2) \textbf{MovieLens-20M}~\cite{DBLP:journals/tiis/HarperK16}, containing 20 million ratings from 138K users on 27K movies. To align with CTR tasks, we binarize ratings (treating $\ge 3$ as clicks), providing a dense, reproducible topology for community benchmarking.
}
\subsubsection{Baselines}
We compare against representative state-of-the-art methods, categorized by modeling paradigm:

\begin{itemize}[left=0pt]
  \item \textbf{Traditional Interaction Models}:  FFM~\cite{DBLP:conf/recsys/JuanZCL16}, DeepFM~\cite{DBLP:conf/ijcai/GuoTYLH17}, AutoInt~\cite{DBLP:conf/cikm/SongS0DX0T19}, and DCNv2~\cite{DBLP:conf/www/WangSCJLHC21}, \change{which explicitly model structured feature interactions. Additionally, we adopt a conventional Embe-dding-MLP framework (denoted as DeepCTR) to serve as a strong baseline for unconstrained implicit interaction learning.}
  \item \textbf{Scaling-Oriented Architectures}: HiFormer~\cite{DBLP:journals/corr/abs-2311-05884}, Wukong~\cite{DBLP:conf/icml/ZhangLCNLLZHYWP24}, HSTU~\cite{DBLP:conf/icml/ZhaiLLWLCGGGHLS24}, and RankMixer~\cite{DBLP:conf/cikm/ZhuFZJWHDWZGYCC25}, representing recent advances in scalable and adaptive model design.
\end{itemize}


\subsubsection{\change{Evaluation Protocol and Implementation Details}}
To enable fair and insightful comparisons, we adopt a \textbf{dual-scale evaluation protocol} aligned with the architectural categories defined above:
\begin{itemize}[left=0pt]
    \item \textbf{Traditional Interaction Models}: Evaluated at their typical capacity of \textbf{$\sim$50M parameters}. We perform rigorous Bayesian search on the validation set to optimize hyperparameters, including embedding dimension, hidden sizes, dropout (0.1--0.5), and L2 regularization ($1e\!-\!6$ to $1e\!-\!3$), to ensure peak performance. Correspondingly, we instantiate \textbf{FAT-Small} ($\sim$50M) for direct comparison.
    
    \item \textbf{Scaling-Oriented Architectures}: Scaled uniformly to approximately \textbf{0.5B parameters} by adjusting width or depth while preserving core architectural constraints (e.g., expert counts in RankMixer or hierarchy levels in HSTU). Hyperparameters are re-tuned under this fixed capacity regime. These are matched against \textbf{FAT-Large} ($\sim$0.5B), alongside \textbf{FAT-XL} ($\sim$1.5B) for scaling trend analysis.
\end{itemize}

\noindent\textbf{Training Infrastructure.}
All experiments are conducted on a distributed training system with 128 NVIDIA GPUs using synchronous data-parallel SGD. \changenew{We use a batch size of 2048 per GPU} and models are optimized using the Adam optimizer ($\beta_1 = 0.9, \beta_2 = 0.999$) with initial learning rates tuned in $\{1e\!-\!4, 3e\!-\!4, 5e\!-\!4, 1e\!-\!3\}$ and a weight decay of $1e\!-\!6$.






\noindent\textbf{Common Settings.} All models share identical feature preprocessing: categorical features are hashed, numerical features are discretized, and sequential behaviors are processed via a shared DIN-style extractor. Embedding dimensions are fixed at $d \in \{8,16,32,64\}$ for different features across all models.


\noindent\textbf{FAT-Specific Configuration.} We set the number of attention heads to 8. The field meta-embeddings $\phi_f \in \mathbb{R}^{64}$ are randomly initialized and shared across layers. The hypernetwork employs $M=64$ shared basis matrices with $\text{Top-}K=3$ sparse activation. Field-pair interaction scalars $w_{f_i,f_j}$ are initialized from $\mathcal{N}(0, 0.01)$.

\begin{table*}[t]
\centering
\caption{\change{Overall CTR performance comparison with state-of-the-art methods under the dual-scale evaluation protocol.}}
\vspace{-1em}
\label{tab:main_results_rq1}
\renewcommand{\arraystretch}{1.02}
\resizebox{1.9\columnwidth}{!}{
\begin{threeparttable}
\begin{tabular}{L{3cm}L{2.5cm}R{2cm}R{2cm}R{2cm}R{2cm}R{2cm}}
\toprule
\multirow{2}{*}{\textbf{Type}} & \multirow{2}{*}{\textbf{Method}} & \multirow{2}{*}{\textbf{\#Params}\tnote{$\dagger$}} & \multicolumn{2}{c}{\textbf{Taobao}} & \multicolumn{2}{c}{\textbf{MovieLens-20M}} \\
\cmidrule(lr){4-5} \cmidrule(lr){6-7}
 & & & \textbf{AUC (\%)} & \textbf{$\Delta$ (\%)} & \textbf{AUC (\%)} & \textbf{$\Delta$ (\%)} \\
\midrule
\textbf{Baseline} & DeepCTR & \textit{40M} & 77.62 & - & 80.12 & - \\
\midrule
\multirow{5}{*}{\textbf{Traditional}}
 & DeepCTR-Large & \textit{0.48B} & 77.64 & +0.02 & 80.16 & +0.04 \\
 & FFM & \textit{25M} & 77.32 & -0.30 & 79.85 & -0.27 \\
 & DeepFM  & \textit{45M} & 77.47 & -0.15 & 80.22 & +0.10 \\
 & AutoInt  & \textit{57M} & 77.67 & +0.05 & 80.48 & +0.36 \\
 & DCNv2  & \textit{47M} & 77.68 & +0.06 & 80.53 & +0.41 \\
\midrule
\multirow{4}{*}{\textbf{Scaling-Oriented}}
 & Wukong  & \textit{0.54B} & 77.73 & +0.11 & 81.02 & +0.90 \\
 & HSTU  & \textit{0.54B} & 77.72 & +0.10 & 81.22 & +1.10 \\
 & HiFormer& \textit{0.58B} & 77.79 & +0.17 & 80.82 & +0.70 \\
 & RankMixer & \textit{0.51B} & 77.85 & +0.23 & 81.62 & +1.50 \\
\midrule
\multirow{3}{*}{\textbf{Ours}}
 & FAT-Small & \textit{52M} & 77.78 & {+0.16} & 80.94 & +0.82 \\
 & FAT-Large & \textit{0.54B} & {78.09} & {{+0.47}} &{83.84} & {+3.72} \\
 & FAT-XL & \textit{1.5B} & \textbf{78.20} & \textbf{{+0.58}} & \textbf{84.50} & \textbf{+4.38} \\
\bottomrule
\end{tabular}
\begin{tablenotes}[flushleft]
\footnotesize
\item[$\dagger$] \changenew{\textbf{\#Params:} Number of dense (non-embedding) parameters.}
\end{tablenotes}
\vspace{-10pt}
\end{threeparttable}
}
\end{table*}

\subsection{Main Results (RQ1)}
\label{sec:rq1_results}

Guided by the evaluation protocol defined above, we present the comparative results in Table~\ref{tab:main_results_rq1}. Our analysis aims to verify whether FAT's performance gains are robust across different capacity regimes, thereby disentangling architectural superiority from simple parameter inflation.

\smallskip

\noindent\textbf{Efficacy in the Small-Scale Regime ($\sim$50M).}
In the typical industrial setting limited to $\sim$50M parameters, \textbf{FAT-Small} consistently outperforms traditional interaction models. 
A critical observation is the \textit{performance saturation} of these baselines: scaling DeepCTR from 40M to 0.48B parameters (DeepCTR-Large) yields negligible improvement (+0.02\%). This indicates that traditional architectures cannot effectively utilize additional capacity. 
In contrast, FAT's superior performance (\textbf{+0.16\%}) under the exact same parameter constraints demonstrates exceptional \textit{parameter efficiency}, proving that its field-aware attention decomposition captures interactions more effectively than implicit MLP structures.


\noindent\textbf{Superiority in the Large-Scale Regime ($\sim$0.5B).}
When scaled to $\sim$0.5B parameters, \textbf{FAT-Large} surpasses modern scaling-oriented architectures. Since all comparison models are uniformly scaled and hyperparameter-tuned, FAT’s persistent lead confirms that its advantage derives from \textit{superior inductive bias} rather than capacity alone.
This structural advantage is even more pronounced on dense topologies: on MovieLens-20M, FAT-Large achieves a massive \textbf{+3.72\%} improvement. We attribute this to FAT's ability to explicitly model dense, high-value semantic connections (e.g., User-Item), whereas standard global attention tends to dilute these specific signals amidst noise.


\noindent\textbf{Scaling Potential.}
Furthermore, extending the model to \textbf{1.5B} parameters (\textbf{FAT-XL}) yields continued performance gains without signs of saturation. This contrasts sharply with the stagnation observed in baselines and offers a preliminary validation of FAT's favorable scaling properties (analyzed further in Section~\ref{sec:scaling}).


\begin{table}[t]
\centering
\caption{Ablation study on FAT components. Values show relative AUC improvement (\%) over the baseline (DeepCTR).}
\label{tab:ablation_rq2}
\renewcommand{\arraystretch}{1.15}
\resizebox{.99\columnwidth}{!}{
\begin{threeparttable}
\begin{tabular}{L{7.5cm}P{2cm}}
\toprule
\textbf{Variant} & \textbf{$\Delta$}\textbf{AUC (\%)} \\
\midrule
{\textbf{FAT-Large}} & \textbf{{+0.47}} \\
\midrule
w/o Field-Aware Biases {($\boldsymbol{b}_{f_i}=\textbf{0}$)} & {+0.41} \\
w/o Interaction Modulation ($w_{f_i,f_j} = 1$) & {+0.35}\\
{w/o Field-Aware Content Alignment (Global $\boldsymbol{W}_{Q,K,V}$)} &  {+0.32}\\
{w/o Field-Aware FFN (Global $\boldsymbol{W}_{1,2}$)} & {+0.35} \\
\midrule
{w/o Hypernetwork (Independent Matrices)} & {+0.46}\tnote{$\dagger$} \\
{w/o Field-Decomposed Attention (Full Pairwise Matrices)} & {N/A}\tnote{$\ddagger$} \\
\bottomrule
\end{tabular}
\begin{tablenotes}[flushleft]
\footnotesize
\item[$\dagger$] {$5\times$ more parameters.}
\item[$\ddagger$] {Out of memory (OOM), exceeding $150$B parameters.}
\end{tablenotes}
\end{threeparttable}
}
\vspace{-25pt}
\end{table}

\vspace{-1em}
\subsection{Ablation Study (RQ2)} 
\label{sec:ablation}

\subsubsection{\change{Impact of Field-Aware Components}}

To understand the sources of FAT’s gains, we conduct ablation studies measuring the \emph{relative AUC improvement} over the baseline in Table~\ref{tab:ablation_rq2}. All variants maintain similar training setups, enabling controlled comparison of design choices.
FAT-Large achieves the highest gain ({+0.47\%}), validating the effectiveness of our overall design. \change{We analyze the contribution of each module following the hierarchy in Table~\ref{tab:ablation_rq2}:}

\begin{itemize}[leftmargin=*,noitemsep]
     \item \textbf{Impact of Structural Priors (w/o Field-Aware Biases):} Removing field-aware positional biases results in a slight performance drop to {+0.41\%}. This indicates that while assigning static semantic roles to fields is beneficial, it serves primarily as a foundational prior rather than the main driver of expressivity.

    \item \textbf{Impact of Interaction Routing (w/o Interaction Modulation):} When we disable the field-pair scalars (setting $w_{f_i,f_j}=1$), the gain drops significantly to {+0.35\%}. This confirms that simply allowing all fields to interact is insufficient; the model requires the ability to explicitly up-weight meaningful pathways (e.g., User $\to$ Item) and suppress noise via asymmetric routing.

    \item \textbf{Impact of Content Alignment (w/o Field-Aware Content Alignment):} 
    \changenew{Replacing the field-specific content projections $\{\mathbf{W}_Q^{(f)}, \mathbf{W}_K^{(f)}, \mathbf{W}_V^{(f)}\}$ with global shared matrices causes the most drastic performance degradation}, dropping to {+0.32\%}, a \textbf{loss of 0.15 percentage points}. This identifies \emph{Field-Aware Content Alignment} as the most critical component, demonstrating that \textit{early specialization by field role} is a prerequisite for modeling deep  interactions.
    
    \item \textbf{\change{Impact of Feature Adaptation (w/o Field-Aware FFN):}} \change{Similarly, reverting the Field-Aware FFN to a standard shared FFN reduces the gain to \textbf{+0.35\%}. This highlights that modeling \emph{intra-field} statistical distributions is as important as modeling \emph{inter-field} interactions, ensuring that the distinct feature spaces of heterogeneous fields are respected.}
\end{itemize}

\subsubsection{\change{Analysis of Basis-Composed Hypernetwork}}

Beyond expressivity, we evaluate the Hypernetwork's role in decoupling model complexity from field cardinality. We first examine its parameter efficiency and necessity:
\begin{itemize}[leftmargin=*,noitemsep]
    \item \textbf{\change{Hypernetwork Efficiency:}} Replacing the Basis-Composed Hypernetwork with fully materialized parameter matrices yields a comparable gain of \textbf{+0.46\%}, but at the cost of $5\times$ more parameters. This proves that our hypernetwork effectively compresses the parameter space without compromising representational power.
    \item \textbf{\change{Necessity of Decomposition:}} Conversely, a naïve attempt to implement field-pair attention without decomposition results in an \textbf{Out-of-Memory (OOM)} error (>150B parameters), underscoring that FAT's structural decomposition is essential for feasibility in industrial systems.
\end{itemize}

\begin{figure}[t]
\centering
\begin{tikzpicture}
\begin{axis}[
font=\scriptsize,
label style={font=\scriptsize},
tick label style={font=\scriptsize},
legend style={font=\scriptsize},
    xlabel={$\boldsymbol{K}$},
    ylabel={$\mathbf{\Delta}\textbf{AUC (\%)}$},
    ylabel style = {yshift=-0.6cm},
    xlabel style = {yshift=0.26cm},
    legend pos=north east,
    legend cell align=left,
    grid=major,
    width=8.5cm,
    height=3.9cm,
    xtick={1,3,5,10},
    ytick={0.2,0.3,0.4, 0.5},
    mark size=1pt,
]

\addplot[color=blue3, mark=triangle*, thick] coordinates {
    (1, 0.447)
    (3, 0.516)
    (5, 0.355)
    (10, 0.275)
};
\addlegendentry{$M=128$}

\addplot[color=blue2, mark=square*, thick] coordinates {
    (1, 0.344)
    (3, 0.47)
    (5, 0.321)
    (10, 0.252)
};
\addlegendentry{$M=64$}
\addplot[color=blue1, mark=*, thick] coordinates {
    (1, 0.172)
    (3, 0.367)
    (5, 0.287)
    (10, 0.241)
};
\addlegendentry{$M=32$}

\end{axis}
\end{tikzpicture}
\vspace{-10pt}
    \caption{\change{Performance spectrum of Basis-Composed Hypernetwork: Trade-off between $M$ and $K$.}}
    \vspace{-5pt}
    \label{fig:km_analysis}
\end{figure}

\begin{figure}[t]
\centering
\begin{tikzpicture}
\begin{axis}[
    width=8.5cm,
    height=3.9cm,
    grid=major,
    ylabel style = {yshift=-0.6cm},
    xlabel style = {yshift=0.26cm},
    xlabel={\changenew{Training Day$(t)$}},
    ylabel={\changenew{$\Delta\text{AUC}(t)$ (\%)}},
    legend pos=south east,
    legend cell align=left,
    tick label style={font=\scriptsize},
    label style={font=\bfseries\scriptsize},
    legend style={font=\scriptsize},
    mark size=1pt,
]

\addplot[color=color3, smooth, thick] coordinates {
    (1,-0.124) (2,-0.009) (3,0.048) (4,0.081) (5,0.099) (6,0.110) (7,0.117)
    (8,0.122) (9,0.125) (10,0.127) (11,0.129) (12,0.131) (13,0.132) (14,0.133)
};
\addlegendentry{FAT}

\addplot[color=color2, smooth, thick] coordinates {
    (1,-0.162) (2,-0.112) (3,-0.069) (4,-0.032) (5,0.002) (6,0.029) (7,0.053)
    (8,0.069) (9,0.082) (10,0.093) (11,0.100) (12,0.107) (13,0.112) (14,0.115)
};
\addlegendentry{FAT w/o Hypernetwork}

\addplot[color=color1, smooth, thick] coordinates {
    (1,-0.153) (2,-0.102) (3,-0.058) (4,-0.022) (5,0.009) (6,0.036) (7,0.057)
    (8,0.073) (9,0.084) (10,0.093) (11,0.099) (12,0.104) (13,0.107) (14,0.109)
};
\addlegendentry{Baseline}
\end{axis}
\end{tikzpicture}
\vspace{-10pt}
\caption{Convergence analysis of \changenew{$\Delta \text{AUC}(t)$} over 14 days following the introduction of a new feature field.}
\vspace{-10pt}
\label{fig:cold_start}
\end{figure}


\noindent\textbf{Parameter Synthesis Mechanism ($\boldsymbol{M}$ vs. $\boldsymbol{K}$).}
To further understand the mechanism of parameter synthesis, we investigate the interplay between the number of bases ($M$) and the \changenew{{active basis count ($K$)}}. As visualized in Figure~\ref{fig:km_analysis}, we observe a distinct coupling relationship between model capacity and sparsity:

\begin{itemize} [left=0pt]
\item \textbf{Impact of Basis Capacity ($\boldsymbol{M}$):} \change{Increasing $M$ consistently yields performance gains (e.g., $M=128$ outperforms $M=32$). This confirms that a richer basis pool expands the representational manifold. However, we cap $M$ in practice to maintain a low memory footprint.}
\item \textbf{Trade-off in Active Basis Count ($\boldsymbol{K}$):} \change{$K$ does not follow a "larger is better" rule; instead, it acts as a lever between expressivity and sparsity. An extremely small $K$ restricts diversity, while an excessively large $K$ leads to over-smoothing.}
\item \textbf{Dynamic Equilibrium:} \change{We identify an inverse relationship between optimal $K$ and $M$. When the basis pool is constrained ($M=32$), a larger $K$ is required to compensate for limited capacity via combinatorial blending. Conversely, as $M$ increases, the optimal $K$ shifts towards $1$. This suggests FAT effectively balances memory and expressivity by trading a small computational cost for exponentially larger representational power.}
\end{itemize}

\smallskip

\noindent\textbf{Efficiency in Feature Evolution.}
\changenew{To validate "sustainable feature evolution" (Section~\ref{sec:hypernetwork}), we first pre-train Baseline, FAT (w/o Hypernetwork), and FAT on the same feature set until convergence. Then, we inject a new field $f_{new}$ and continue training over a 14-day period. \changenew{Here, $\Delta\text{AUC}(t)$ quantifies the AUC gain of each framework relative to its own converged performance before field injection, $\Delta \text{AUC}(t)=\text{AUC}_{\text{old}+f_{new}}(t)-\text{AUC}_{\text{old}}(0)$, where $t$ denotes the training days elapsed since the injection.} As shown in Figure~\ref{fig:cold_start}, {FAT} achieves rapid convergence ("warm start") by leveraging pre-trained shared bases; it only requires learning a lightweight routing signature for the new field. In contrast, {FAT (w/o Hypernetwork)} and Baseline suffer from the "cold start" problem, exhibiting similar sluggish trajectories as they must learn full-rank parameters from scratch. While FAT (w/o Hypernetwork) is expected to eventually match FAT (consistent with Table~\ref{tab:ablation_rq2}), FAT significantly reduces the time-to-value for new features, enabling agile industrial iteration.}

\vspace{-0.5em}
\subsection{Interpretability Analysis (RQ3)} 
\label{sec:interpretability}

To understand how FAT captures semantic interactions, we analyze the learned modulation weights $w_{f_i,f_j}$, which govern the intensity of information flow from field $f_j$ to $f_i$. These parameters are shared across tokens within the same field pair and reflect global interaction patterns. We visualize the average $w_{f_i,f_j}$ across all attention heads in Figure~\ref{fig:w_matrix}, revealing two key properties: \textbf{structured coherence} and \textbf{asymmetric influence}.

\begin{table}[b]
\centering
\vspace{-10pt}
\caption{Asymmetry in $w_{f_i,f_j}$: Item as query attends more selectively than user interest.}
\vspace{-10pt}
\label{tab:asymmetry}
\small
\resizebox{.9\columnwidth}{!}{
\begin{tabular}{llP{1.3cm}P{1.3cm}P{1.3cm}}
\toprule
 & & \multicolumn{3}{c}{\textbf{Key Field} ($f_j$)} \\
\cmidrule(lr){3-5}
& & item & \changenew{recent\_clicks}   & user\_profile \\
\midrule
\multirow{3}{*}{\textbf{Query Field} ($f_i$)} & \multicolumn{1}{|l|}{item}            & {-}     & \textbf{0.97}    & 0.16 \\
& \multicolumn{1}{|l|}{\changenew{recent\_clicks}}      & 0.23  & {-}      & 0.24 \\
& \multicolumn{1}{|l|}{user\_profile}      & 0.24 & 0.32 & {-} \\
\bottomrule
\end{tabular}
}
\end{table}

\noindent\textbf{Structured Coherence.} 
The weight matrix is sparse and exhibits clear block-wise structure. High values concentrate on semantically meaningful pairs:
\begin{itemize}[leftmargin=*,topsep=0pt,itemsep=-1pt]
    \item \changenew{Candidate item features, including category and shop-level attributes (e.g., {\tt item\_cate}, {\tt shop\_level}), exhibit strong links to real-time user signals (e.g., {\tt recent\_clicks}).}
    \item \changenew{User profile fields, including demographic attributes (e.g., \texttt{age}, \texttt{gender}), interact most strongly with long-term preference indicators (e.g., \texttt{fav\_brands}, \texttt{longterm\_clicks}).}
\end{itemize}
In contrast, cross-field interactions between unrelated modalities (e.g., \texttt{device\_type} $\to$ \texttt{income}) remain near zero. This pattern \change{aligns with domain intuition}: short-term intent drives item relevance, while static profiles shape stable preferences. FAT naturally learns this separation without explicit supervision, demonstrating that its structured expressivity focuses capacity on meaningful pathways.

\smallskip

\noindent\textbf{Asymmetric Influence.} 
\changenew{The interaction patterns exhibit strong directionality, reflecting a functional hierarchy between fields. As shown in Table~\ref{tab:asymmetry}, when a candidate \texttt{item} acts as the query, it attends intensely to \texttt{recent\_clicks} (0.97). This indicates that the item operates as an active retriever, aggressively synthesizing user intent from behavioral history to determine its own relevance.
In contrast, when a \texttt{recent\_clicks} feature serves as the query, its attention to the \texttt{item} is weaker and more diffuse (0.23). This suggests that historical behaviors act as passive supporting evidence---they provide context when queried but do not require the candidate item to define their own semantics. This asymmetry confirms that FAT correctly learns the causal structure of recommendation: items seek users, not the other way around.}



\begin{figure}[t]
    \centering
    \includegraphics[width=0.6\linewidth]{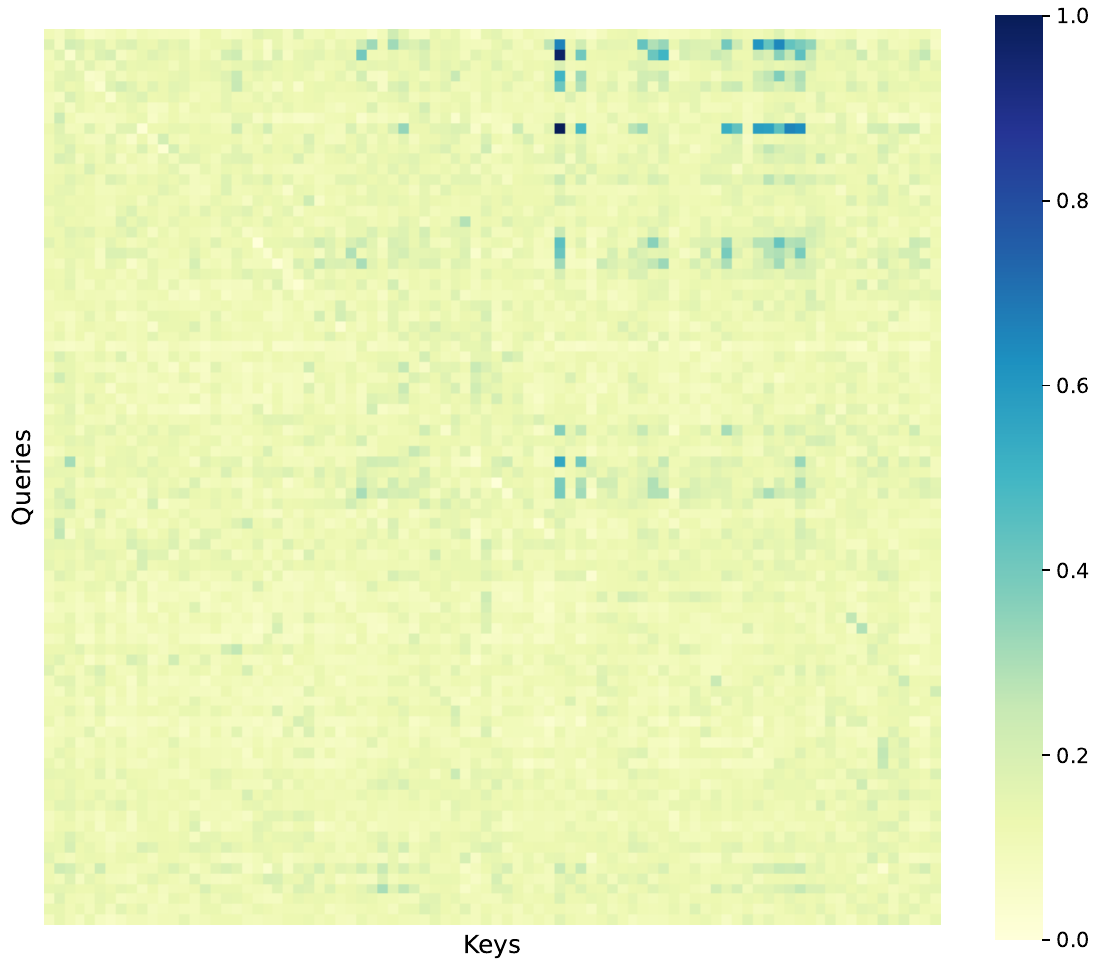}
    \vspace{-10pt}
    \caption{Heatmap of $w_{f_i,f_j}$. Strong interactions (dark) align with expected semantic dependencies.}
    \vspace{-5pt}
    \label{fig:w_matrix}
\end{figure}

\begin{figure}[t]
\centering
\begin{tikzpicture}
\begin{loglogaxis}[
font=\footnotesize,
label style={font=\footnotesize\bfseries},
tick label style={font=\footnotesize},
legend style={font=\footnotesize},
    legend cell align=left,
    width=8.5cm,
    height=4.8cm,
    grid=both,
    grid style={gray!25},
    minor grid style={gray!15},
    ylabel style = {yshift=-0.6cm},
    xlabel={Number of Dense Parameters},
    xlabel style = {yshift=0.2cm},
    ylabel={$\Delta$AUC (\%)},
    legend style={font=\scriptsize, draw=gray!50, fill=white, fill opacity=0.9, text opacity=1},
    legend pos=north west,
    xmin=4e7, xmax=1.6e9,
    ymin=1.3e-1, ymax=6.5e-1,
    ytick={0.2,0.3,0.4,0.6},
    yticklabels={0.2,0.3,0.4,0.6},
]

\addplot[
    only marks,
    mark=*,
    mark size=1pt,
] coordinates {
    (5.2e7, 0.16)
    (8.0e7, 0.17)
    (1.35e8, 0.21)
    (2.02e8, 0.24)
    (2.7e8, 0.31)
    (4.0e8, 0.34)
    (5.4e8, 0.47)
    (8.1e8, 0.46)
    (1.45e9, 0.58)
};
\addlegendentry{FAT (Empirical)}

\addplot[
    red,
    thick,
    dashed,
    domain=4.5e7:1.55e9,
    samples=200
] {1.16e-4 * x^(0.405)};
\addlegendentry{$\Delta\mathrm{AUC}=1.16\times 10^{-4}\cdot N_{\text{params}}^{0.405}$}
\end{loglogaxis}
\end{tikzpicture}
\vspace{-10pt}
\caption{\change{Power-law relationship between parameter count and AUC. Best-fit slope $\Delta \text{AUC} =1.16 \times 10^{-4} \cdot N_{\text{params}}^{0.405}$.}}
\vspace{-12pt}
\label{fig:scaling}
\end{figure}


\vspace{-1em}
\subsection{Scaling Behavior (RQ4)} 
\label{sec:scaling}

We examine whether FAT exhibits predictable scaling behavior, where performance improves systematically with capacity under fixed feature semantics, as suggested by its theoretical generalization bound (Theorem~\ref{thm:generalization}). 
To this end, we evaluate a series of FAT variants scaled from 50M to 1.5B dense (non-embedding) parameters, trained on the same dataset with consistent hyperparameters and infrastructure. As shown in Figure~\ref{fig:scaling}, $\Delta$AUC increases monotonically with parameter count across three orders of magnitude, following a power-law trend. The observed relationship is well-characterized by the empirical function:
\begin{equation}
    \text{$\Delta$AUC} = 1.16 \times 10^{-4} \cdot N_{\text{params}}^{0.405},
    \label{eq:scaling_law}
\end{equation}
 No saturation is observed within the tested range, indicating that FAT continues to benefit from increased capacity---an uncommon property in CTR models, where unstructured architectures often plateau or degrade under scale~\cite{DBLP:conf/icml/ZhangLCNLLZHYWP24,DBLP:conf/icml/ZhaiLLWLCGGGHLS24,DBLP:conf/cikm/ZhuFZJWHDWZGYCC25}.

This smooth scaling behavior aligns with Theorem~\ref{thm:generalization}, which establishes that FAT's complexity grows as $\mathcal{O}(Fd^2 + F^2)$, independent of token vocabulary size $n$. By constraining expressive capacity to field-aware interaction pathways, FAT ensures additional parameters refine meaningful patterns rather than overfitting. Importantly, we clarify that $F$ defines the interaction topology, not a scaling dimension. Unlike standard scaling, increasing $F$ fundamentally alters the input schema. Thus, the observed trend in Figure~\ref{fig:scaling}  reflects \textit{intra-schema scalability}: performance improves predictably as representational fidelity deepens within this fixed structural foundation.

Further ablation in Section \ref{sec:ablation} shows that removing field-aware components compromises this scalability, proving the behavior arises from architectural design.
\change{In summary, FAT provides clear empirical validation for a scaling law in CTR prediction.} This closes the loop between theory and practice, establishing a principled path toward scalable recommendation modeling.

\subsection{Online A/B Test Results (RQ5)} 
\label{sec:ab_test}
To evaluate the real-world impact of FAT, we conducted a large-scale online A/B test on Taobao’s sponsored search system, one of the largest e-commerce recommendation systems globally.
Traffic was evenly split between the \textit{control group}, serving the existing production model (a highly optimized traditional CTR model with manual feature crosses), and the \textit{treatment group}, which replaced only the prediction module with \textbf{FAT-Large} ($\sim$0.5B parameters). 
\change{We evaluate the model performance from two perspectives as reported in Table~\ref{tab:ab_result}: \textbf{(1) System Efficiency}: FAT maintains comparable inference speed to the online baseline: P99 latency increases only marginally (45ms to 48ms), despite a $5\times$ increase in \changenew{model FLOPs (Floating Point Operations)}. \changenew{We attribute this stability to our hardware-aligned model design and optimization strategies: Model FLOPs Utilization (MFU) surges from 5\% to 34\%.}
Without caching, P99 latency increases to 133ms, rendering the model undeployable. This confirms that FAT translates offline complexity into structured expressivity without prohibitive online penalties; \textbf{(2) Business Gain}: FAT achieves a \textbf{+2.33\%} lift in {CTR}, indicating superior user intent capture, and improves {RPM} (Revenue Per Mille impressions) by \textbf{+0.66\%}, translating to substantial daily revenue growth, thereby confirms FAT’s industrial cost-effectiveness.}


\begin{table}[t]
\centering
\caption{{Online A/B test results comparing system efficiency and business gain.}}
\vspace{-10pt}
\label{tab:ab_result}
\renewcommand{\arraystretch}{1.15}
\resizebox{.99\columnwidth}{!}{
\begin{threeparttable}
\begin{tabular}{lR{2cm}R{1.1cm}R{1.3cm}R{1cm}R{1cm}}
\toprule
\multirow{2}{*}{\textbf{Model Setting}} & \multicolumn{3}{c}{\textbf{System Efficiency}} & \multicolumn{2}{c}{\textbf{Business Gain}} \\
\cmidrule(lr){2-4} \cmidrule(lr){5-6}
 & \textbf{P99 Latency} & \textbf{GFLOPs} &\textbf{MFU} &  \textbf{CTR} & \textbf{RPM} \\
\midrule
{Online Baseline} & 45 ms &  20 &5\% & -- & -- \\
\midrule
{FAT-Large} & {48 ms} & {100} &\textbf{34\%} &  {\textbf{+2.33\%}} & {\textbf{+0.66\%}} \\
FAT-Large (\textit{Uncached} $\mathbf{W}_\Phi^{(f)}$) & 133 ms & 150 & {N/A}\tnote{$^\dagger$} &{N/A}\tnote{$^\dagger$} & {N/A}\tnote{$^\dagger$} \\
\bottomrule
\end{tabular}

\begin{tablenotes}[flushleft]
\footnotesize
\item[$\dagger$] {Service Unavailable: Severe timeouts prevented the model from serving live requests. Consequently, performance metrics (e.g., MFU) and downstream business indicators (e.g., CTR, RPM) become invalid due to the inability to sustain traffic.}
\end{tablenotes}
\end{threeparttable}
}
\vspace{-15pt}
\end{table}

\section{Conclusion}
\label{sec:conclusion}

\change{This work addresses the scalability bottleneck in CTR prediction by resolving the structural conflict between sequential Transformers and combinatorial data. 
We propose the \textbf{Field-Aware Transformer (FAT)}, which synergizes {field-aware} components to impose inductive biases aligned with heterogeneous fields, and employs a basis-composed hypernetwork to cut parameter complexity for efficient online deployment.
Theoretical analysis and empirical results converge to validate {a principled scaling law}: performance improves predictably with capacity under this structured design.
Validated in both offline benchmarks and live production, FAT confirms that breaking the curse of diminishing returns requires a paradigm shift---from parameter inflation to \textit{structured expressivity} rooted in domain semantics.
}

\clearpage
\newpage

\appendix

\section{APPENDIX: THEORETICAL ANALYSIS OF FAT}
\label{app:theory}

In this section, we provide a rigorous analysis of the generalization properties of the Field-Aware Transformer (FAT), based on Rademacher complexity~\cite{DBLP:books/daglib/0034861}. We derive an upper bound on the generalization error that explicitly depends on the number of semantic fields $F$, the embedding dimension $d$, and the interaction structure, while being independent of the total vocabulary size $n$. This justifies why FAT scales effectively even under extreme sparsity.

\subsection{Setup and Notation}
Let $\mathcal{X}$ denote the space of CTR inputs, each consisting of $F$ semantic fields $\{f_1, ..., f_F\}$. Consider one attention layer of FAT applied to embedded features $\{\boldsymbol{h}_i\}_{i=1}^N$, where $\boldsymbol{h}_i \in \mathbb{R}^d$ includes both content embedding and field-aware positional bias. The output for token $i$ is:
\begin{align*}
\text{Output}_i &= \sum_{j=1}^{N} \alpha(i, j)\boldsymbol{v}_j, \\
\quad \alpha(i, j) &= \frac{\exp\left(s(i, j)/\sqrt{d}\right)}{\sum_{k=1}^{N} \exp \left( s(i, k)/\sqrt{d} \right)}, \quad s(i, j)=(\boldsymbol{q}_i^\top \boldsymbol{k}_j) \cdot w_{f_i,f_j},
\end{align*}
where:
\begin{itemize}[left=0pt]
\item $\boldsymbol{q}_i= \boldsymbol{h}_i\boldsymbol{W}^{(f_i)}_Q $, $\boldsymbol{k}_j= \boldsymbol{h}_j\boldsymbol{W}^{(f_j)}_K $, $\boldsymbol{v}_j= \boldsymbol{h}_j\boldsymbol{W}^{(f_j)}_V $;
    \item $\boldsymbol{W}^{(f)}_Q \in \mathbb{R}^{d\times d}$, similarly for $\boldsymbol{W}^{(f)}_K ,\boldsymbol{W}^{(f)}_V$;
    \item $w_{f_i,f_j} \in \mathbb{R}$ is a learnable scalar coefficient.
\end{itemize}
Let $\mathcal{F}_{\text{FAT}}$ denote the function class induced by this attention block, mapping input sequences $\boldsymbol{H}$ to contextualized outputs. Our goal is to bound its Rademacher complexity.

\subsection{Assumptions}
We make the following boundedness assumptions, common in generalization analysis:
\begin{assumption}[Bounded Embeddings]
$\|\boldsymbol{h}_i \|_2 \leq R$ for all $i$, some constant $R> 0$.
\end{assumption}
\begin{assumption}[Bounded Parameters]
For all fields $f$:
\begin{itemize}[left=0pt]
	\item $\|\boldsymbol{W}^{(f)}_Q \|_F \leq B_Q$,
    \item $\|\boldsymbol{W}^{(f)}_K \|_F \leq B_K$,
    \item $\|\boldsymbol{W}^{(f)}_V \|_F \leq B_V$;
\end{itemize}
and for all field pairs $(f_i, f_j)$: $|w_{f_i,f_j}| \leq B_w$.
These reflect practical regularization techniques (e.g., weight decay, gradient clipping).
\end{assumption}

\subsection{Main Theorem and Proof}
Our main theoretical result is stated as follows:

\begin{theorem}[Generalization Bound for FAT]
\label{thm:fat_app}
Under Assumptions 1--2, the empirical Rademacher complexity of $\mathcal{F}_{\text{FAT}}$ satisfies:
\[
\hat{\mathfrak{R}}_S(\mathcal{F}_{\text{FAT}}) \leq \mathcal{O}\left( \frac{ \sqrt{F d^2 + F^2} }{\sqrt{m}} \cdot C(R, B_Q, B_K, B_V, B_w, d) \right),
\]
where $C(R, B_Q, B_K, B_V, B_w, d) = \mathcal{O}(R^2 B_Q B_K B_w \sqrt{d} + R B_V B_w)$ is a constant depending on the norm bounds and dimension. Consequently, with probability at least $1-\delta$, the generalization error satisfies:
\[
L_{\text{gen}} \leq L_{\text{train}} + \mathcal{O}\left( \frac{ \sqrt{F d^2 + F^2} }{\sqrt{m}} \cdot C + \sqrt{\frac{\log(1/\delta)}{m}} \right).
\]
\end{theorem}

\begin{proof}
The proof proceeds in three steps, leveraging standard tools from learning theory.

\textbf{Step 1: Bounding Score Magnitude and Lipschitz Constants.}
Consider the score $s(i, j)=(\boldsymbol{q}_i^\top \boldsymbol{k}_j) \cdot w_{f_i,f_j}$. By Assumptions 1 and 2, we have:
\[
\|\boldsymbol{q}_i \|_2 = \|\boldsymbol{h}_i \boldsymbol{W}^{(f_i)}_Q \|_2 \leq \|\boldsymbol{h}_i \|_2 \|\boldsymbol{W}^{(f_i)}_Q \|_F  \leq B_Q R,
\]
similarly $\|\boldsymbol{k}_j \|_2 \leq B_K R$. Therefore,
\[
|\boldsymbol{q}_i^\top \boldsymbol{k}_j| \leq \|\boldsymbol{q}_i \|_2 \|\boldsymbol{k}_j \|_2 \leq B_Q B_K R^2.
\]
It follows that $|s(i, j)| \leq B_Q B_K R^2 B_w \triangleq C_s$. The softmax function $\sigma(\cdot)$, when viewed as a map from logits in $\mathbb{R}^N$ to probabilities in $[0,1]^N$, is $C_s$-Lipschitz with respect to the $\ell_\infty$ norm on the input and $\ell_1$ norm on the output, provided the logit differences are bounded by $C_s$~\cite{DBLP:books/daglib/0034861}.

\textbf{Step 2: Complexity of the Score Function Class.}
Define the class of score functions:
\[
\mathcal{S} = \left\{ (i, j) \mapsto (\boldsymbol{h}_i \boldsymbol{W}^{(f_i)}_Q )^\top (\boldsymbol{h}_j \boldsymbol{W}^{(f_j)}_K ) \cdot w_{f_i,f_j} \right\}.
\]
We analyze the Rademacher complexity of $\mathcal{S}$. The key insight is that the parameters defining $\mathcal{S}$ can be grouped into two types: (a) the $3F$ matrices ($\boldsymbol{W}^{(f)}_Q, \boldsymbol{W}^{(f)}_K, \boldsymbol{W}^{(f)}_V$) with Frobenius norm constraints, and (b) the $F^2$ scalars $w_{f_i,f_j}$ with absolute value constraints.

Using the fact that the Rademacher complexity of a sum of function classes is bounded by the sum of their complexities, and applying standard results for linear function classes under norm constraints~\cite[Chapter 3]{DBLP:books/daglib/0034861}, we get:
\begin{align*}
\hat{\mathfrak{R}}_S(\mathcal{S}) \leq \underbrace{\mathcal{O}\left( \frac{B_Q B_K R^2 \sqrt{F d^2}}{\sqrt{m}} \right)}_{\text{from } \boldsymbol{W}^{(f)}_Q, \boldsymbol{W}^{(f)}_K} + \underbrace{\mathcal{O}\left( \frac{B_Q B_K R^2 B_w \sqrt{F^2}}{\sqrt{m}} \right)}_{\text{from } w_{f_i,f_j}} \\
= \mathcal{O}\left( \frac{C_s \sqrt{F d^2 + F^2}}{\sqrt{m}} \right).
\end{align*}

\textbf{Step 3: Propagating Complexity to the Output.}
The final output is a weighted sum of the values $\boldsymbol{v}_j = \boldsymbol{h}_j \boldsymbol{W}^{(f_j)}_V $. Each value vector satisfies $\|\boldsymbol{v}_j \|_2 \leq B_V R$.

The attention weights $\alpha(i,j)$ are a function of the scores $\{s(i,k)\}_{k=1}^N$. Since the softmax is $C_s$-Lipschitz, we can apply Talagrand's contraction lemma~\cite[Theorem 4.12]{DBLP:books/daglib/0034861} to relate the Rademacher complexity of the attention weights to that of the scores. Furthermore, because the output is a linear combination of the values weighted by these attention weights, and given the boundedness of the values, we can combine these results.

After applying the (vector) contraction principle to propagate complexity through the softmax and the subsequent weighted sum, we obtain an additional dimension-dependent factor when controlling the $\ell_2$-norm of the $d$-dimensional outputs; together with the $1/\sqrt{d}$ temperature scaling (absorbed into constants), this leads to the stated $C(\cdot)$ term.

After careful application of the contraction lemma and accounting for the $\sqrt{d}$ factor introduced by the temperature scaling in the softmax (which affects the Lipschitz constant), we arrive at the final bound:
\[
\hat{\mathfrak{R}}_S(\mathcal{F}_{\text{FAT}}) \leq \mathcal{O}\left( \frac{ \sqrt{F d^2 + F^2} }{\sqrt{m}} \cdot (R^2 B_Q B_K B_w \sqrt{d} + R B_V B_w) \right).
\]
Applying the standard generalization bound via Rademacher complexity~\cite[Theorem 3.3]{DBLP:books/daglib/0034861} completes the proof.
\end{proof}

\subsection{Discussion: From Tight Generalization to Predictable Scaling}
\label{sec:app_scaling}

The derived generalization bound reveals that the generalization gap of FAT scales as $\tilde{\mathcal{O}}(\sqrt{F d^2 + F^2}/\sqrt{m})$, which crucially depends only on the number of semantic fields $F$ and the embedding dimension $d$, not on the total vocabulary size $n$. This structural alignment dramatically improves statistical efficiency compared to standard Transformers, whose effective hypothesis space complexity implicitly scales with $n$ due to unstructured attention over all tokens.

However, a tight generalization bound alone does not guarantee a smooth \textit{scaling law}; it ensures that the model won't overfit, but not that it will continue to improve. For predictable performance gains as capacity increases, two conditions must hold simultaneously:
\begin{enumerate}
    \item \textbf{Stable Generalization (Low Variance):} The generalization gap must remain controlled. This is guaranteed by Theorem~\ref{thm:fat_app}.
    \item \textbf{Enhanced Expressivity (Reduced Bias):} The model's representational capacity must grow in a way that allows it to capture increasingly complex patterns in the data, thereby reducing the training error $L_{\text{train}}$.
\end{enumerate}

FAT satisfies the second condition through its \textit{structured expressivity}. As the embedding dimension $d$ increases:
\begin{itemize}[left=0pt]
\item The field-aware projections $\boldsymbol{W}^{(f)}_Q, \boldsymbol{W}^{(f)}_K, \boldsymbol{W}^{(f)}_V$ gain higher rank, enabling more nuanced, field-specific encoding of content.
    \item The interaction modulation factors $w_{f_i,f_j}$, though scalars, operate on richer, higher-dimensional representations, allowing for more sophisticated modeling of cross-field dependencies.
\end{itemize}
This architecture ensures that additional parameters are used to refine meaningful, semantically grounded interactions rather than fitting noise. Empirically, this manifests as a systematic reduction in $L_{\text{train}}$ as $d$ grows.

Combining these two effects---controlled variance from the structural prior and reduced bias from enhanced expressivity---we achieve a synergistic improvement in test performance $L_{\text{gen}}$. Under a fixed data distribution and field schema $F$, if we scale the model width such that $N_{\text{params}} \propto F d^2$, then the generalization gap shrinks as $\mathcal{O}(\sqrt{N_{\text{params}}}/\sqrt{m})$. When $m$ grows proportionally to $N_{\text{params}}$, this gap remains stable or decreases, while $L_{\text{train}}$ continues to decrease. 
This synergy enables the observed power-law scaling behavior. 
Let $\Delta \text{AUC}$ be the gain over a fixed baseline. Our experiments (Section~\ref{sec:scaling}) validate that:
\begin{equation}
\Delta \text{AUC} \propto N_{\text{params}}^{\beta}, \quad \beta > 0,
\end{equation}
holds across multiple orders of magnitude. This is the first theoretically grounded scaling law for CTR models, demonstrating that predictable scaling arises from principled architectural design that aligns with the data's combinatorial semantics.

\end{document}